\documentclass[journal]{IEEEtran}

\usepackage[T1]{fontenc}
\usepackage[utf8]{inputenc}
\usepackage[pdftex]{graphicx}
\graphicspath{{./figures}}
\usepackage{cite}
\usepackage{amsthm,amsmath,amssymb,amsfonts}
\usepackage{float}
\usepackage{epstopdf}
\usepackage{enumitem}
\usepackage{tikz}
\usetikzlibrary{shapes,arrows,positioning,calc,fit}
\usepackage{booktabs}
\usepackage{siunitx}

\newtheorem{assumption}{Assumption}

\newtheorem{remark}{Remark}
\newtheorem{theorem}{Theorem}

\usepackage{xcolor} \usepackage{hyperref}
\hypersetup{
    colorlinks=true,
    linkcolor={blue!50!black},
    citecolor={blue!50!black},
    urlcolor={blue!50!black}
}

\usepackage{fancyhdr}

\newcommand{\R}{\mathbb{R}} 
\newcommand{\N}{\mathbb{N}} 
\newcommand{\Ni}[2]{\N_{#1}^{#2}} 
\newcommand{\cc}[1]{{\mathcal{#1}}} 

\newcommand{\T}{\top} 
\newcommand{\becauseof}[2][=]{\stackrel{\scriptstyle\mkern-1.5mu#2\mkern-1.5mu}{#1}}

\def\cP{\cc{P}} 
\def\cS{\cc{S}} 
\def\cC{\cc{C}} 
\def\nx{{n_x}} 
\def\nu{{n_u}} 
\def\nz{{n_z}} 
\def\ny{{n_y}} 
\def\cX{\cc{X}} 
\def\cY{\cc{Y}} 
\def\cZ{\cc{Z}} 
\def\cU{\cc{U}} 
\def\cO{\cc{O}} 
\def\ck{\kappa} 
\def\cD{\cc{D}} 
\def\nth{{n_\theta}} 
\def\nxi{{n_\xi}} 

\begin{document}

\title{\LARGE \bf Learning practically stabilizing output-feedback nonlinear controllers}

\author{Kui~Xie,~Pablo~Krupa,~Alberto~Bemporad%
\thanks{
This work was funded by the European Union (ERC Advanced Research Grant COMPACT, No. 101141351). Views and opinions expressed are however those of the authors only and do not necessarily reflect those of the European Union or the European Research Council. Neither the European Union nor the granting authority can be held responsible for them.
The authors are with the IMT School for Advanced Studies, Lucca, Italy.
Corresponding author: Kui~Xie.
Emails: {\tt \{kui.xie, pablo.krupa, alberto.bemporad\}@imtlucca.it}.
}
}

\pagestyle{fancy}
\maketitle
\thispagestyle{fancy}

\begin{abstract}
This paper addresses the problem of learning an output-feedback surrogate controller offline that approximates a given, possibly computationally expensive, nonlinear controller-observer pair.
The surrogate is modeled as a recurrent dynamical system and is trained to imitate closed-loop input/output trajectories generated by the given controller.
Beyond imitation accuracy, the offline training problem promotes input-to-state practical stability by incorporating estimated state trajectories to learn a candidate Lyapunov function.
The approach is validated on a nonlinear continuous stirred tank reactor, where constraint satisfaction and practical stability are assessed through a probabilistic validation approach.
The numerical results highlight the benefit of jointly learning the Lyapunov function by comparing against an imitation-only baseline.
\end{abstract}

\begin{IEEEkeywords}
Controller learning, nonlinear model predictive control, extended Kalman filtering, Lyapunov stability.
\end{IEEEkeywords}

\section{Introduction} \label{sec:intro}

Nonlinear model predictive control (NMPC), combined with state estimators, provides a principled output-feedback control framework to handle nonlinear dynamics, constraints, and tracking requirements in advanced control applications~\cite{RMD17}. In particular, controller-observer architectures such as NMPC coupled with an extended Kalman filter (EKF) are a popular choice when the plant state is not directly measured~\cite{LR94,HPB12}. Their practical deployment, however, may be hindered by the computational burden of repeatedly updating the EKF and solving the nonlinear NMPC optimization problem online. This is especially relevant in embedded and fast-sampling applications.

These considerations have motivated a broad literature on reducing the online complexity of MPC.
A classical alternative is explicit MPC, which removes the need for online optimization by precomputing the control law offline, but is typically limited to relatively small systems and linear or piecewise-affine prediction models~\cite{AB09,RMD17}.
As an alternative, a large literature has investigated {\it approximate} MPC representations.
Approximate MPC seeks to substitute the exact MPC control law with a computationally cheap function,
such as a neural network (NN), returning similar outputs for each given current state and reference inputs~\cite{CSA18, HKT18, KL19, MDW20, KL20, KAL21, FG23, TH24, ATS+24, HKZ25, TFZ+25}.
As an additional alternative, the authors in~\cite{LL24} propose learning a NN imitating an iterative solver that converges to the KKT conditions of the NMPC optimization problem.

Several efforts were devoted to guaranteeing closed-loop properties of approximate MPC, for example by combining robust MPC design with bounds or statistical validation of the controller approximation error~\cite{HKT18,ATS+24,HKZ25,TFZ+25}, computing worst-case approximation errors of explicit MPC through mixed-integer~\cite{FG23} and global nonlinear optimization~\cite{Bem26b}, or validating closed-loop performance probabilistically~\cite{KL19,KAL21}.
We note that deterministic guarantees provided by the aforementioned articles require knowledge of the system state, which is not available in an NMPC-EKF setting. 
On the other hand, probabilistic guarantees are provided in~\cite{KL19,KAL21} when considering the use of an estimated state, at the expense of requiring a large number of closed-loop experiments using the learned approximate MPC controller.

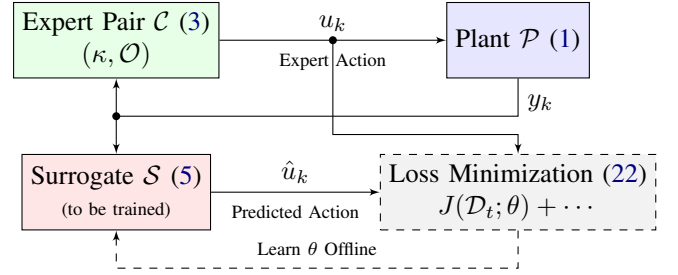
\begin{figure}[t!]
    \centering
    \begin{tikzpicture}[auto, node distance=2cm, >=latex',
        block/.style = {draw, rectangle, minimum height=1cm, minimum width=1.5cm, align=center, fill=white},
        sum/.style = {draw, circle, inner sep=1pt, node distance=1.5cm},
        input/.style = {coordinate},
        tmp/.style = {coordinate}
    ]
    
    
        \node [block, fill=blue!10] (plant) {Plant $\mathcal{P}~\eqref{eq:plant}$};
        
        \node [block, fill=green!10, left=3.0cm of plant] (expert) {Expert Pair $\cC$~\eqref{eq:expert}\\$(\ck,\cO)$};
        
        \node [block, fill=red!10, below=1.0cm of expert] (surrogate) {Surrogate $\cS~\eqref{eq:surr_control}$\\ \scriptsize (to be trained)};
        
        \node [block, fill=gray!10, dashed] (loss) at (plant |- surrogate) {Loss Minimization~\eqref{eq:learn_prob_stabilize}\\$J(\cD_t; \theta) + \cdots$};
    
        
        \draw [->] (expert) -- node[name=u] {$u_k$} coordinate[midway] (u_wire) (plant);
        
        \draw [->] (plant.south) -- node[midway, right, yshift=-2pt] {$y_k$} ++(0,-0.5) -| coordinate (y_point) (expert.south);
        
        
        \draw [->] (y_point) -- (surrogate.north);
        \fill (y_point) circle (1.5pt); 
    
        \draw [->] (surrogate) -- node [name=u_hat_node] {$\hat{u}_k$} (loss);
    
        \draw [->] (u_wire) -- ++(0,-1.25) -| (loss.north);
        \fill (u_wire) circle (1.5pt); 
    
        \draw [->, dashed] (loss.south) -- ++(0,-0.5) -| node[pos=0.25, above=0.05cm] {\scriptsize Learn $\theta$ Offline} (surrogate.south);
    
        \node [below=0.05cm of u] {\scriptsize Expert Action};
        \node [below=0.05cm of u_hat_node] {\scriptsize Predicted Action};
        
    \end{tikzpicture}
    \caption{The surrogate controller $\cS$ is trained to imitate the input/output behavior of the expert controller $\ck$ and observer $\cO$ (e.g., NMPC+EKF). The training involves minimizing the error $J(\cD_t; \theta)$ between the expert control input $u_k$ and the surrogate output $\hat{u}_k$, and other terms, given the same measurement history~$y_k$.}
    \label{fig:surrogate_control}
\end{figure}

In this paper we consider an approximate MPC approach that directly takes the system outputs as the input to the surrogate controller, bypassing the need for state estimation.
Instead of approximating the NMPC control law as a static map from state to control input, we propose to train a dynamic surrogate controller from input/output trajectories to imitate the observer-controller function from output measurements to NMPC control actions. The general surrogate-learning setup is illustrated in Fig.~\ref{fig:surrogate_control}. 

Beyond imitation accuracy, we train the surrogate controller to promote practical stability properties
of the closed-loop system.
We propose to jointly learn a candidate Lyapunov function serving as a stability-oriented regularizer during surrogate training. This connects the proposed approach to learning Lyapunov functions to certify stability of nonlinear systems in closed-loop with NN-based controllers~\cite{RBK18,AAG21}.
In contrast to these works, our purpose is not to certify the learned surrogate over the full domain, but to regularize a recurrent output-feedback surrogate controller learned from the closed-loop input/output trajectories obtained by an expert controller-observer pair. We consider a relaxed non-monotone Lyapunov condition~\cite{AP08} to provide additional flexibility while retaining a simple candidate Lyapunov function.
We show that satisfaction of this non-monotone Lyapunov descent condition leads to input-to-state practical stability (ISpS)~\cite{LAR+09} under bounded measurement noise.

The main contributions are as follows: \emph{(i)} we formulate surrogate learning directly at the input/output level for nonlinear controller-observer pairs, thereby addressing output-feedback implementations in which the plant state is unavailable; \emph{(ii)} we propose a stability-oriented training objective that combines imitation and Lyapunov-based regularization through a relaxed two-step decrease condition; and \emph{(iii)} we show that this two-step decrease condition leads to ISpS.
Additionally, we provide numerical results in which we learn an NMPC-EKF surrogate controller for a nonlinear continuous stirred-tank reactor (CSTR).
We highlight the effectiveness of the proposed approach using probabilistic validation~\cite{KAL21}, showing that we can jointly learn a surrogate controller and a candidate Lyapunov function that guarantee ISpS with high probability.

The remainder of the paper is organized as follows. Section~\ref{sec:problem} formulates the control problem. Section~\ref{sec:regulation} presents the output-feedback surrogate and the stabilization-oriented learning framework, and provides the conditions for ISpS. Section~\ref{sec:results} presents the CSTR case study.
Section~\ref{sec:conclusions} concludes the paper.

\subsubsection*{Notation}
The set of real (non-negative) numbers is denoted by $\R$ ($\R_{\geq 0}$).
The set of natural numbers is denoted by $\N$, with $\Ni{a}{b}\doteq\{a,a+1,\ldots,b\}$ for $a \le b$.
Vector inequalities are componentwise. We denote by
$I_n$ and $^\top$, respectively, the identity matrix and transposition.
A function $\alpha \colon \R_{\geq 0} \to \R_{\geq 0}$ is of class $\cc{K}$ if it is continuous, strictly increasing and $\alpha(0) = 0$; and is of class $\cc{K}_\infty$ if $\alpha \in \cc{K}$ and $\alpha(a) \to \infty$ as $a \to \infty$.
For $v \in \R^n$, $\|v\|_2$ and $\|v\|_\infty$ denote the Euclidean and infinity norms, respectively, while $\|v\|_Q^2 \doteq v^\top Qv$ for a symmetric $Q$.

\section{Problem formulation} \label{sec:problem}

Consider a plant $\cP$ described by the discrete-time dynamics
\begin{equation} \label{eq:plant}
\cP:\, \begin{cases}
\begin{aligned}
    x_{k+1} &= f_\cP(x_k, u_k), \\
    y_k     &= g_\cP(x_k) + \eta_k, 
\end{aligned}
\end{cases}
\end{equation}
where $x_k \in \R^{\nx}$, $u_k \in \R^{\nu}$, and $y_k \in \R^{\ny}$ are, respectively, the state, input, and output of the plant at time step $k$.
The measurement noise $\eta_k \in \R^{\ny}$ is assumed to form an i.i.d.\ sequence.
We assume that $\cP$ is observable and controllable, that $y_k$ can be measured, and that $x_k$ is not available.

The control objective is to steer the output $y_k$ to a given reference $\bar{y} \in \R^{n_y}$, while satisfying state and input constraints $x_k \in \cX$ and $u_k \in \cU$, $\forall k$, where $\cX \subset \R^{\nx}$ is a compact set and $\cU \doteq \{u \in \R^{n_u} \colon \underline{u} \leq u \leq \overline{u}\}$, with $\underline{u} < \overline{u}$.
We assume that $\bar y$ admits a unique admissible steady-state pair $(x_s,u_s)\in\cX\times\cU$ defined by
\begin{equation} \label{eq:steady_state}
\bar{y} = g_\cP(x_s), \quad x_s = f_\cP(x_s, u_s).
\end{equation}

We consider the existence of an \emph{expert controller} $\cC$ that has been tuned to provide a certain desired closed-loop dynamics but is computationally expensive to evaluate in real time.
In particular, we consider $\cC$ to be given by
\begin{equation} \label{eq:expert}
\cC:\, \begin{cases}
\begin{aligned}
    \chi_{k} &= \cO(\chi_{k-1}, y_k, u_{k-1}),\quad \hat{x}_k = h(\chi_{k}), \\
    u_k &= \ck(\hat{x}_k, \bar{y}),
\end{aligned}
\end{cases}
\end{equation}
where $\cc{O} \colon \R^{n_\chi} \times \R^{n_y} \times \R^{n_u} \to \R^{n_\chi}$ is an observer whose internal state $\chi_k$ is updated using the output measurement $y_k$ and past control input $u_{k-1}$, and provides the estimated state~$\hat{x}_k$ of~\eqref{eq:plant} through some function $h \colon \R^{n_\chi} \to \R^{n_x}$; and $\ck \colon \R^{n_x} \times \R^{n_y} \to \cU$ is a (nonlinear) control law that provides the input $u_k$.

Additionally, we assume that $\cP$ is initially operated using some \emph{initialization procedure},
i.e., by an \emph{initialization controller}
\begin{equation} \label{eq:init_control}
\cC_0:\, \begin{cases}
\begin{aligned}
    \chi_{k} &= \cO(\chi_{k-1}, y_k, u_{k-1}),\quad \hat{x}_k = h(\chi_{k}), \\
    u_k &= \ck_0(\hat{x}_k),\quad k=-M,-M+1,\ldots,-1,
\end{aligned}
\end{cases}
\end{equation}
where $\ck_0 \colon \R^{n_x} \to \cU$ is the \emph{initialization control law}
and $M>0$.
We define $k = 0$ the sample time when $\cC_0$ is switched to $\cC$.
We assume that the initialization procedure drives the system to some unknown state $x_0 \in \cX_0$, where $\cX_0 \subseteq \cX$ is an unknown compact set with $x_s \in \cX_0$. 
The task of the \emph{expert controller} is to drive the system from $x_0$ to the reference
steady state $x_s$.
We note that the use of an \emph{initialization procedure}, whose control law differs from the one used during normal operations, is a typical approach in many real-life applications~\cite{HSS07,MHDIS12}.
One of the purposes of this initialization procedure is to ensure that the predicted state $\hat{x}_0$ provided by the observer $\cO$ is a reasonable estimate of the real system state $x_0$ when the controller is switched to $\ck$.

A classical example of a computationally expensive control architecture~\eqref{eq:expert} is an NMPC-EKF controller, where the internal state $\chi_k$ is formed by the estimated state of the system and the \emph{a posteriori} estimate of the covariance matrix, and the evaluation of $\ck$ is done by solving a (nonlinear) finite-horizon constrained optimal control problem, cf. Section~\ref{sec:results:expert_architecture}.

We assume that $\cC$ has been designed to satisfy the control objective of steering $y_k$ toward $\bar{y}$ while satisfying the system constraints.
We do not assume that the domain of attraction to $x_s$ provided by $\cC$ is the entire set $\cX_0$ necessarily.
Instead, we assume that it has been tuned to provide a closed-loop behavior that is \emph{desirable} and that generally satisfies the control objectives.
That is, in practice we have that $\chi_k$ converges to a fixed value $\chi_\infty \in \R^{n_\chi}$ and that $y_k$ converges to a bounded set around $\bar{y}$ as $k \to \infty$.
We provide a loose definition of a \emph{desirable} closed-loop behavior of the controller $\cC$
on purpose, as it is problem-dependent.
As motivation, note that in a practical setting it may be difficult to design an NMPC-EKF controller with stability guarantees, e.g., computing a suitable terminal invariant set for the NMPC controller may be intractable.
However, an NMPC-EKF without formal stability guarantees may still lead to practical stability and constraint satisfaction.

\section{Approximate controller-observer} \label{sec:regulation}

In this section we present a method to learn a surrogate controller of $\cC$ jointly with a Lyapunov function for the resulting closed-loop system.
The idea is to learn a computationally cheap input/output surrogate controller that provides practical stability toward $\bar{y}$, even when the expert controller $\cC$ is not perfectly stabilizing by design.
In the following, we present the surrogate controller and the stabilization-oriented learning problem, and we show that satisfaction of the learned two-step non-monotone Lyapunov descent condition implies ISpS.

\subsection{Surrogate controller of a controller-observer pair}
We consider a \emph{surrogate controller} $\cS$ given by
\begin{equation} \label{eq:surr_control}
\cS:\,
    \begin{cases}
        \begin{aligned}
            z_{k+1} &= f_z(z_k,y_k,\bar{y};\theta), \\
            \hat{u}_{k} &= f_u(z_k,y_k,\bar{y};\theta),
        \end{aligned}
    \end{cases}
\end{equation}
where $z_k\in \R^{n_z}$ is the internal state, $\hat{u}_k \in \R^{\nu}$, and $\theta \in \R^{\nth}$ collects the parameters defining $f_z\colon \R^{\nz} \times \R^{\ny} \times \R^{\ny} \to \R^{\nz}$ and $f_u\colon \R^{\nz} \times \R^{\ny} \times \R^{\ny} \to \R^{\nu}$.
The control law of $\cS$ is $u_k = \hat{u}_k$.
The structure of the surrogate controller~\eqref{eq:surr_control} is motivated by the expert controller-observer architecture~\eqref{eq:expert}.
In particular,~\eqref{eq:surr_control} has an internal state $z_k$, the measured output $y_k$ as its input, and the control action $\hat{u}_k$ as its output.

We propose a particular structure for the functions of the surrogate controller~\eqref{eq:surr_control}, where we restrict $f_z$ and $f_u$ to satisfy~
\begin{subequations} \label{eq:surr_out_sets}
\begin{align}
    &f_z(z,y,\bar y;\theta)\in\cZ, \label{eq:surr_out_sets:z} \\
    &f_u(z,y,\bar y;\theta)\in\cU, \label{eq:surr_out_sets:u}
\end{align}
\end{subequations}
$\forall \theta \in \R^{\nth}$, $y \in \R^{n_y}$, $z \in \cZ$, with $\cZ \doteq \{z\in\R^{\nz} \colon \| z \|_\infty \leq \zeta \}$, for some scalar $\zeta > 0$.
We impose~\eqref{eq:surr_out_sets:u} so that the control law satisfies the input constraints, and~\eqref{eq:surr_out_sets:z} to limit the internal states of the surrogate controller.
Let $v_k \doteq (z_k, y_k, \bar{y}) \in \R^{n_v}$, $n_v = n_z + n_y + n_y$.
We impose~\eqref{eq:surr_out_sets} by taking
\begin{subequations} \label{eq:surr_structure}
   \begin{align}
       f_z(v_k; \theta) &= \zeta \tanh \left(\hat{f}_z(v_k; \theta_z) \right), \label{eq:nn_fz}\\
       f_u(v_k; \theta) &= \min\{\overline{u},\, \max\{\underline{u}, \, \hat{f}_u(v_k; \theta_u)\}\}, \label{eq:nn_fu}
   \end{align} 
\end{subequations}
where $\theta = (\theta_z, \theta_u)$, with $\theta_z$ the parameters of $\hat{f}_z \colon \R^{n_v} {\to \R^{n_z}}$, $\theta_u$ the parameters of $\hat{f}_u \colon \R^{n_v} \to \R^{n_u}$, and the hyperbolic tangent function $\tanh$ is taken componentwise. Since $z_k$ has no physical meaning, since now on we take for convenience $\zeta = 1$, i.e., $\cZ = [-1, 1]^{n_z}$.

The objective is to learn $\theta$ so that, under the same measurement history, the surrogate output $\hat{u}_k$ is as close as possible to the expert input $u_k$, thereby approximating the closed-loop input/output behavior provided by the expert controller $\cC$.
This is illustrated by Fig.~\ref{fig:surrogate_control}.
Starting from some (unknown) $x_0 \in \cX_0$, we initialize $z_0 = 0$ as a reasonable arbitrary initialization procedure satisfying $z_0 \in \cZ$.
Our objective is for the output $y_k$ of the closed-loop system to converge to the reference $\bar{y}$.
Therefore, $(x_k, u_k)$ should converge to $(x_s, u_s) \in \cX \times \cU$ satisfying~\eqref{eq:steady_state}.
To achieve this, the internal state $z_k$ of $\cS$ must converge to some $z_s \in \cZ$ satisfying
\begin{equation}\label{eq:equilibrium}
    z_s = f_z(z_s,\bar y,\bar y;\theta),
    \quad
    u_s = f_u(z_s,\bar y,\bar y;\theta),
\end{equation}
so that if $y_k = \bar{y}$, $x_k = x_s$ and $z_k = z_s$, then $x_{k+1} = x_s$.
We note that the requirement that $z_k$ converges to some fixed point $z_s$ is analogous to the internal state $\chi_k$ of the observer $\cO$ of the expert controller~\eqref{eq:expert} converging to some $\chi_\infty$.

The closed-loop system formed by~\eqref{eq:plant} and~\eqref{eq:surr_control} is given by
\begin{subequations} \label{eq:cl_explicit}
\begin{align}
    x_{k+1} &= f_\cP\!\left(x_k, f_u(z_k, y_k, \bar{y};\theta)\right), \\
    z_{k+1} &= f_z(z_k, y_k, \bar{y};\theta), \\
    y_k &= g_\cP(x_k) + \eta_k,
\end{align}
\end{subequations}
For convenience, we denote by $\xi_k \doteq (x_k,z_k) \in \R^{\nxi}$, $\nxi = n_x + n_z$, and $\tilde{x}_k \doteq x_k - x_s, \; \tilde{z}_k \doteq z_k - z_s, \; \tilde{\xi}_k \doteq \xi_k - \xi_s,$ where $\xi_s \doteq (x_s, z_s)$. Then,~\eqref{eq:cl_explicit} can be written in shifted coordinates as:
\begin{equation} \label{eq:cl}
    \tilde{\xi}_{k+1} = f_\cS(\tilde{\xi}_k, \eta_k; \theta) 
    \doteq F_\cS(\tilde{\xi}_k + \xi_s, \eta_k; \theta) - \xi_s,
\end{equation}
where $F_\cS(\xi_k, \eta_k; \theta)$ denotes the closed-loop dynamics in the original (non-shifted) coordinates.

To encourage practical stability, we propose jointly learning $\theta$ and a Lyapunov function for the closed-loop dynamics~\eqref{eq:cl}~\cite{KB59,Kha02}.
In the present imitation-learning setting, we use a relaxed two-step Lyapunov condition~\cite{AP08} to provide flexibility when using a simple structure for the candidate Lyapunov function.
We seek a scalar $\tau \geq 0$ and uniformly continuous functions $V, S\colon \R^\nxi \to \R_{\geq 0}$ satisfying
\begin{equation} \label{eq:Lyap_cond_2}
    \tau V(\tilde{\xi}_{k+2}) + V(\tilde{\xi}_{k+1}) - (1+\tau) V(\tilde{\xi}_k) \leq - S(\tilde{\xi}_k)
\end{equation}
for all extended states $\xi_k$ belonging to trajectories of~\eqref{eq:cl} starting from $z_0 = 0$ and with states $x_0 \in \cX_0$ obtained from the initialization procedure~\eqref{eq:init_control}.
When $\tau=0$, \eqref{eq:Lyap_cond_2} reduces to the standard one-step monotonic Lyapunov descent, cf.~\cite[Appendix B.3]{RMD17}.
Positive values of $\tau$ permit temporary increases in $V$ while still enforcing an overall two-step decrease.

In the absence of measurement noise, i.e., $\eta_k = 0$, $\forall k$, it is easy to verify that the two-step Lyapunov descent condition~\eqref{eq:Lyap_cond_2} can be equivalently expressed using a standard single-step monotonic Lyapunov descent condition.
Indeed, let
\begin{equation} \label{eq:one_step_Lyap_W}
    W(\tilde{\xi}) \doteq (1+\tau) V(\tilde{\xi}) + \tau V(f_\cS(\tilde{\xi},0; \theta)).
\end{equation}
Then, denoting $W_k \doteq W(\tilde{\xi}_k)$, $V_k \doteq V(\tilde{\xi}_k)$ and $S_k \doteq S(\tilde{\xi}_k)$,
\begin{equation*}
    W_{k+1} - W_k = \tau V_{k+2} + V_{k+1} - (1 + \tau) V_k \leq - S_k \leq 0,
\end{equation*}
showing that satisfaction of~\eqref{eq:Lyap_cond_2} implies $\tilde{\xi}_k \to 0$ as ${k \to \infty}$, cf.~\cite[Theorem 2.1]{AP08}.
The following theorem extends this to the presence of measurement noise $\eta_k$, showing that satisfaction of a two-step Lyapunov descent condition implies ISpS; see~\cite[Definition 6]{LAR+09}.
For convenience, we define
\begin{equation}
    \Gamma_k \doteq \tau V_{k+2} + V_{k+1} - (1+\tau) V_k + S_k,
\end{equation}
where $\tilde{\xi}_{k+1} = f_\cS(\tilde{\xi}_k, \eta_k; \theta)$ and $\tilde{\xi}_{k+2} = f_\cS(\tilde{\xi}_{k+1}, \eta_{k+1}; \theta)$.

\begin{assumption}[Conditions for ISpS] \label{ass:ISpS}
Assume that: 
\begin{enumerate}[label=(\textit{\roman*}),itemsep=0pt,topsep=0pt]
    \item \label{ass:ISpS:K-infty} There exist $\alpha_1,\alpha_2,\alpha_3\in\cc{K}_\infty$ such that
    \begin{equation} \label{eq:ISpS_bounds}
        \alpha_1(\|\tilde{\xi}\|)\leq V(\tilde{\xi})\leq \alpha_2(\|\tilde{\xi}\|), \quad
        S(\tilde{\xi})\geq \alpha_3(\|\tilde{\xi}\|).
    \end{equation}
    \item \label{ass:ISpS:Lipschitz} Functions $f_\cP$, $g_\cP$, $f_z$ and $f_u$ are Lipschitz continuous, and~\eqref{eq:equilibrium} is satisfied.
    \item $\eta_k \in \cc{W}_\eta$, $\forall k$, where $\cc{W}_\eta \subset \R^{n_y}$ is a compact set containing the origin.
\end{enumerate}
\end{assumption}

\begin{theorem}[ISpS of the surrogate controller] \label{theo:ISpS}
Let Assumption~\ref{ass:ISpS} be satisfied.
Let $\Omega = \{ \{\tilde{\xi}_k\}_{k = 0}^{\infty} : \tilde{\xi}_k \in \text{trajectory of } \eqref{eq:cl}, \text{ with } x_0 \in \cX_0, z_0 = 0, \eta_k \in \cc{W}_\eta \}$.
Assume that $x_k \in \cX$, $\forall \{\tilde{\xi}_k\}_{k=0}^{\infty} \in \Omega$.
Then~\eqref{eq:cl} is ISpS if there exist a function $\alpha_\eta \in \cc{K}$ and a scalar $c \geq 0$ such that
\begin{equation} \label{eq:V_ISpS_cond}
\Gamma_k \leq \alpha_\eta(\|(\eta_{k+1}, \eta_k)\|) + c, \quad \forall \{\tilde{\xi}_k\}_{k=0}^{\infty} \in \Omega.
\end{equation}
\end{theorem}

\begin{proof}
By Assumption~\ref{ass:ISpS}\ref{ass:ISpS:Lipschitz}, $f_\cS(0,0;\theta)=0$ and $f_\cS$ is Lipschitz continuous; let $L_f$ be the Lipschitz constant of $f_\cS$.
Let
\begin{equation} \label{eq:one_step_Lyap_W_noise}
    W(\tilde{\xi}; \eta) \doteq (1+\tau) V(\tilde{\xi}) + \tau V(f_\cS(\tilde{\xi},\eta; \theta)),
\end{equation}
be the evaluation of~\eqref{eq:one_step_Lyap_W} conditioned on $\eta \in \cc{W}_\eta$.
We prove that satisfaction of~\eqref{eq:V_ISpS_cond} implies ISpS by showing that $W(\tilde{\xi}; \eta)$ is an ISpS-Lyapunov function for $f_\cS$, see~\cite[Definition 7]{LAR+09}.
$\forall \tilde{\xi} \in \Omega$, $\forall \eta \in \cc{W}_\eta$, we have that $W(\tilde{\xi}; \eta) \geq (1 + \tau)\alpha_1(\|\tilde{\xi}\|)$ 
by~\eqref{eq:ISpS_bounds}, and that
\begin{align*}
    W(\tilde{\xi}; \eta)
    &\becauseof[\leq]{(*)} (1 + \tau)\alpha_2(\|\tilde{\xi}\|) + \tau\alpha_2(\|f_\cS(\tilde{\xi},\eta;\theta)\|) \\
    &\becauseof[\leq]{(\star)} (1 + \tau)\alpha_2(\|\tilde{\xi}\|) + \tau\alpha_2(L_f\|(\tilde{\xi}, \eta)\|) \\
    &\becauseof[\leq]{(\dagger)} (1 + \tau)\alpha_2(\|\tilde{\xi}\|) + \tau\alpha_2(L_f\|\tilde{\xi}\|+L_f\|\eta\|) \\
    &\becauseof[\leq]{(\dagger)} (1 + \tau)\alpha_2(\|\tilde{\xi}\|) + \tau\alpha_2(2 L_f\|\tilde{\xi}\|) + \tau \alpha_2(2 L_f\|\eta\|) \\
    &\becauseof[\leq]{(\ddagger)} (1 + \tau)\alpha_2(\|\tilde{\xi}\|) + \tau c_\sigma \sigma(\|\tilde{\xi}\|) + \tau c_\sigma \sigma(\|\eta\|)
\end{align*}
where $(*)$ follows from~\eqref{eq:ISpS_bounds},
$(\star)$ follows from $f_\cS(0, 0; \theta) {=} 0$ and Lipschitz continuity of $f_\cS$, 
$(\dagger)$ from the triangle inequality,
and the existence of $c_\sigma > 0$ and $\sigma \in \cc{K}_\infty$ in step $(\ddagger)$ from~\cite[Corollary 10]{S98}.
Since $\cc{W}_\eta$ is bounded and $\sigma \in \cc{K}_\infty$, $\exists c_1 \geq 0 \colon \tau c_\sigma \sigma(\|\eta\|) \leq c_1, \forall \eta \in \cc{W}_\eta$.
Using $S(\tilde{\xi}) \geq \alpha_3(\|\tilde{\xi}\|)$ in~\eqref{eq:ISpS_bounds}, ISpS then follows from satisfaction of~\eqref{eq:V_ISpS_cond}~\cite[Thm. 3]{LAR+09}.
\end{proof}

\subsection{Learning a surrogate controller with practical stability}\label{sec:stability_prob}
To learn $\theta$, $V$ and $S$, we perform $M_t$ experiments obtained from plant~\eqref{eq:plant} in closed-loop with~\eqref{eq:expert} under the constant reference $\bar{y}$.
Each trajectory starts from an unknown $x_0 \in \cX_0$ as a result of the initialization procedure~\eqref{eq:init_control}.
We collect the input-output data of each experiment, as well as the estimated states returned by the observer $\cO$, in the dataset
\begin{equation}
    \cD_t \doteq \{ \{(u_k^j, y_k^j, \hat{x}^j_k)\}_{k=0}^{N_t^j-1}\}_{j=1}^{M_t},
\end{equation}
where $N_t^j > 1$ is the length of each experiment $j \in \N_1^{M_t}$.

\vspace*{0.2em}\noindent\textbf{Imitation objective:}
For a given $\theta$, we associate to each $\{y_k^j\}_{k = 0}^{N_t - 1}$ in $\cD_t$ sequences $\{z_k^j, \hat{u}_k^j\}_{k=0}^{N_t^j - 1}$, satisfying~
\begin{subequations} \label{eq:predicted_surr}
    \begin{align}
        z_0^j &= 0,  \label{eq:predicted_surr:z0} \\
        z_{k+1}^j &= f_z(z_k^j,y_k^j,\bar{y};\theta), \; k \in \Ni{0}{N_t^j-2}, \label{eq:predicted_surr:zk} \\
        \hat{u}_k^j &= f_u(z^j_k, y^j_k,\bar{y};\theta), \; k \in \Ni{0}{N_t^j-1}, \label{eq:predicted_surr:u}
    \end{align}
\end{subequations}
predicting the evolution of the internal state and control input of the surrogate controller for each observed output sequence.
The objective is for $\{\hat{u}_k^j\}$ to match $\{u_k^j\}$ when provided the output measurement sequence $\{y_k^j\}$, so that the surrogate controller~$\cS$ provides the same result as the expert controller~$\cC$.
We can pose this as a learning problem whose imitation loss has the form of a system identification problem based on input/output data~\cite{SL19,Bem23,Bem25}.
In particular, we use the mean-squared imitation loss
\begin{equation*}
    J(\cD_t) = \frac{1}{M_t} \sum_{j=1}^{M_t} \frac{1}{N_t^j} \sum_{k=0}^{N_t^j-1} \|u^j_{k}- \hat{u}_k^j \|_2^2.
\end{equation*}

\noindent\textbf{Lyapunov descent condition:}\label{sec:lyap_reg}
We consider quadratic functions $V$ and $S$ given by
\begin{subequations} \label{eq:quad_V_S}
\begin{align}
    V(\tilde{\xi}) &= \| \tilde{\xi} \|_P^2 + \mu_V \| \tilde{\xi} \|_2^2, \\
    S(\tilde{\xi}) &= \| \tilde{\xi} \|_Q^2 + \mu_S \| \tilde{\xi} \|_2^2 + \beta_x \| \tilde{x} \|^2_2 + \beta_z \| \tilde{z} \|^2_2,
\end{align}
\end{subequations}
with $P = \psi_V^\T \psi_V$, $Q = \psi_S^\T \psi_S$, where $\psi_V \in \R^{n_V \times n_\xi}$, $\psi_S \in \R^{n_S \times n_\xi}$, and $\mu_V, \mu_S > 0$ are the trainable parameters, while $\beta_x, \beta_z > 0$ are non-trainable parameters that are included to encourage strict decrease away from~$\xi_s$.
This parametrization guarantees $P,Q \succeq 0$, which along with $\mu_V,\mu_S, \beta_x, \beta_z>0$ yields functions $V$ and $S$ satisfying Assumption~\ref{ass:ISpS}\ref{ass:ISpS:K-infty}\footnote{E.g., take $\alpha_1(r)=\mu_V r^2$, $\alpha_2(r)=(\lambda_{\max}(P)+\mu_V)r^2$, and $\alpha_3(r)=\mu_S r^2$, where $\lambda_{\max}(P)$ is the maximum eigenvalue of~$P$.}.

Since the real system states along the experiments are unknown, we enforce the non-monotone Lyapunov descent condition~\eqref{eq:Lyap_cond_2} on the estimated states $\hat{x}_k^j$ available in $\cD_t$, as a proxy for the true states $x_k^j$.
Let $\hat{\xi}_k^j = (\hat{x}_k^j, z_k^j)$, and denote
$\hat{\xi}_{k+1|k}^j \doteq F_\cS(\hat{\xi}_k^j, 0; \theta)$, $\hat{\xi}_{k+2|k}^j \doteq F_\cS(\hat{\xi}_{k+1|k}^j, 0; \theta)$
the successor states at the next two sample times for the extended state $\hat{\xi}_k^j$.
Finally, let
\begin{equation}\label{eq:Gamma}
    \begin{aligned}
        \hat \Gamma_k^j \doteq \, & \tau V(\hat{\xi}_{k+2|k}^j - \xi_s) +V(\hat{\xi}_{k+1|k}^j - \xi_s)  \\
                             & - (1+\tau)V(\hat{\xi}_k^j - \xi_s) + S(\hat{\xi}_k^j - \xi_s).
    \end{aligned}
\end{equation}
We encourage satisfaction of the non-monotone Lyapunov descent condition~\eqref{eq:Lyap_cond_2} using the loss
\begin{equation} \label{eq:loss_lyap}
    J_L(\cD_t; \psi, \mu, \tau) = \alpha_L \frac{1}{M_t} \sum_{j = 1}^{M_t} \frac{1}{N_t^j} \sum_{k = 0}^{N_t^j - 1} \max\{ 0, \hat \Gamma_k^j \},
\end{equation}
where $\psi = (\psi_V, \psi_S)$ and $\mu = (\mu_V, \mu_S)$ are the parameters of functions $V$ and $S$, and $\alpha_L > 0$ is a penalty hyper-parameter.

\noindent\textbf{Equilibrium consistency:} To promote the equilibrium consistency~\eqref{eq:equilibrium} of the surrogate with the steady-state pair $(x_s,u_s)$, we add the further penalty term
\begin{equation*} \label{eq:loss_equilibrium}
    J_s(z_s)
    {=} \alpha_s^z \| f_z(z_s, \bar{y}, \bar{y}; \theta) - z_s \|_2^2
    {+} \alpha_s^u \| f_u(z_s, \bar{y}, \bar{y}; \theta) - u_s \|_2^2,
\end{equation*}
where $\alpha_s^z, \alpha_s^u > 0$ are penalty hyper-parameters.

\noindent\textbf{Overall training problem:} Combining all the penalties introduced above, we obtain the following optimization problem for learning a practically stabilizing surrogate controller:
\begin{subequations} \label{eq:learn_prob_stabilize}
\begin{align}
    \min_{\Theta}\;& J(\cD_t; \theta) + J_L(\cD_t; \psi, \mu, \tau) + J_s(z_s) + r_\Theta(\Theta)
    \label{eq:learn_prob_stabilize:cost} \\
    \textrm{s.t.}\;&~\eqref{eq:predicted_surr}, \quad j \in \Ni{1}{M_t}, \\
    \;&~\mu_V \ge \underline{\mu}_V,\quad \mu_S \ge \underline{\mu}_S,\quad \tau \ge 0,\quad z_s \in \cZ,
\end{align}
\end{subequations}
where $\Theta \doteq (\theta, \psi, \mu, \tau, z_s)$ collects all trainable parameters of the surrogate controller~\eqref{eq:surr_control} and the Lyapunov descent condition~\eqref{eq:Lyap_cond_2}, $\underline{\mu}_V, \underline{\mu}_S > 0$ are small scalars used to enforce $\mu_V, \mu_S > 0$, and $r_\Theta(\Theta)$ is a regularization term of the trainable parameters (typically choices are the $\ell_2$ or $\ell_1$ norm).

Note that, if functions $f_z$, $f_u$, $f_\cP$ and $g_\cP$ are automatically differentiable (AD), problem~\eqref{eq:learn_prob_stabilize} can be solved by standard machine learning tools and optimization methods, such as Adam~\cite{KB14} or L-BFGS-B~\cite{BLN95}.

\begin{remark}
We chose quadratic functions $V$ and $S$ in~\eqref{eq:quad_V_S} for convenience of exposition. In general, any functions satisfying Assumption~\ref{ass:ISpS}\ref{ass:ISpS:K-infty} are valid.
\end{remark}

\begin{remark}\label{rem:lyap_subset}
In practice, we find that it is not necessary to include all $N_t^j$ steps in~\eqref{eq:loss_lyap}.
Instead, the horizons $N_t^j$ in~\eqref{eq:loss_lyap} may be replaced by $N_L^j \doteq \min\{N_t^j,N_L\}$, where $N_L$ is a hyper-parameter chosen according to the settling time of the expert closed-loop system, so that the trajectories have reached a vicinity of the target $\bar{y}$. This retains the informative transient portion of the data while avoiding unnecessary Lyapunov evaluations near steady state. The cost can be further reduced by evaluating~\eqref{eq:loss_lyap} only on a subset $M_s \leq M_t$ of the trajectories, selected by greedy farthest-point sampling of the initial estimates~\cite{XB24}, e.g., choose $i_1 \in \Ni{1}{M_t}$ and $i_t \in \arg\max_{j \in \Ni{1}{M_t}\setminus\{i_1,\ldots,i_{t-1}\}} \min_{r \in \Ni{1}{t-1}} \|\hat{x}_0^j-\hat{x}_0^{i_r}\|_2$ for $t=2,\ldots,M_s$.
\end{remark}

\section{Numerical results} \label{sec:results}
\subsection{Plant description} \label{sec:results:plant}
To validate the proposed methodology, we consider the CSTR benchmark from~\cite[pp.~559--563]{Beq98}.
The system dynamics describe the evolution of the state vector 
$x_k = [T_k, C_{A,k}]^\top$, where $T_k$ [\SI{}{K}] is the reactor temperature and $C_{A,k}$ [\SI{}{\kilo\mol\per\meter\cubed}] is the concentration of reactant A at sample time $k$. 
The model parameters are taken from Case~2 in~\cite[p.~563]{Beq98}. 
We obtain~\eqref{eq:plant} by applying Euler integration, with a sampling time of $T_s=\SI{0.5}{\hour}$, on the following continuous-time dynamics:
\begin{align*}
        \frac{d T}{d t} &= \frac{F}{V}(T_f - T(t)) + \frac{-\Delta H}{\rho c_p} r(t) - \frac{UA}{V\rho c_p}(T(t) - T_j(t)), \\
        \frac{d C_{A}}{d t} &= \frac{F}{V}(C_{Af} - C_{A}(t)) - r(t),
\end{align*}
where the reaction rate is $r(t) = k_0 e^{- \Delta E/RT_k}C_{A}(t)$ and the output of the system is $C_{A}$, i.e., $g_\cP(x_k) = C_{A, k}$.
The control input is the jacket temperature $T_{j}(t)$, constrained to
$T_{j}(t)\in[285.15,312.15]$.
The state-constraint set is $\cX \doteq \{x=[T,C_A]^\top : T\in[310,400],\, C_A\in[0,8]\}$. 
The objective is to steer the output $C_A$ to the target $\bar{y} = 3.5$, whose steady-state~\eqref{eq:steady_state} is given by $x_s = (361.0, 3.5)$ and $u_s = 292.2$.

\subsection{Expert NMPC-EKF controller}\label{sec:results:expert_architecture}

We take an NMPC and EKF as the expert controller $\ck$ and observer $\cO$~\eqref{eq:expert}.
In particular, we consider the NMPC problem~
\begin{subequations} \label{eq:NMPC}
\begin{align}
    \min_{U_k} \quad & \sum_{i=0}^{H-1} \left( \|y_{k+i|k} - \bar{y}\|^2_{W_y} + \|u_{k+i|k} - u_s\|^2_{W_u} \right) \\
    \text{s.t.} \quad & \hat{x}_{k+i+1|k} = f_\cP(\hat{x}_{k+i|k}, u_{k+i|k}), \; i \in \Ni{0}{H-1}, \\
    & y_{k+i|k} = g_\cP(\hat{x}_{k+i|k}), \; i \in \Ni{0}{H-1}, \\
    & u_{k+i|k} \in \cU, \; i \in \Ni{0}{H-1}, \\
    & \hat{x}_{k+i|k} \in \cX, \; i \in \Ni{0}{H},
\end{align}
\end{subequations}
where $U_k = \{u_{k|k}, \dots, u_{k+H-1|k}\}$ is the sequence of control inputs over the prediction horizon $H$, and $\hat{x}_{k+i|k}$ is the predicted state at time $k+i$ based on information at time $k$. The weighting matrices $W_y \succeq 0$ and $W_u \succ 0$ define the tracking and input penalties, and the initial state of the prediction, $\hat{x}_{k|k}$, is provided by the EKF observer $\cO$.
Consistently with the timing convention in~\eqref{eq:expert}, given $u_{k-1}$ and $y_k$, the standard EKF updates are given by
\begin{subequations}
\begin{align}
    \hat{x}_{k|k-1} &= f_\cP(\hat{x}_{k-1|k-1}, u_{k-1}), \\
    P_{k|k-1} &= A_{k-1} P_{k-1|k-1} A_{k-1}^\top + Q_x, \\
    K_k &= P_{k|k-1} C_k^\top (C_k P_{k|k-1} C_k^\top + R_y)^{-1}, \\
    \hat{x}_{k|k} &= \hat{x}_{k|k-1} + K_k \bigl(y_k - g_\cP(\hat{x}_{k|k-1})\bigr), \\
    P_{k|k} &= (I - K_k C_k) P_{k|k-1}(I - K_k C_k)^\top + K_k R_y K_k^\top,
\end{align}
\label{eq:EKF}
\end{subequations}
where $A_{k} \doteq \nabla_x f_\cP(\hat{x}_{k|k}, u_{k})$ and $C_k \doteq \nabla_x g_\cP(\hat{x}_{k|k-1})$ are the Jacobians of the system dynamics and output function. The matrices $Q_x$ and $R_y$ denote the process and measurement noise covariances, respectively, and the pair $(\hat{x}_{k|k},P_{k|k})$ forms the observer state~$\chi_k$.

For the NMPC, we take the prediction horizon $H = 10$, $W_y = 1$ and $W_u = 0.1$. For the EKF, we take $Q_x = 10^{-2} I$ and $R_y = 1$. 
The measurement noise $\eta_k$ is sampled from a zero-mean Gaussian distribution with standard deviation $0.02$, truncated to the interval $[-0.04,0.04]$.

\subsection{Initialization controller} \label{sec:results:initialization}

For convenience and ease of exposition, we take the controller~$\ck_0$ of the initialization procedure~\eqref{eq:init_control} as the same NMPC~\eqref{eq:NMPC} and EKF~\eqref{eq:EKF}.

The initialization procedure is taken as follows.
First, a random state is drawn from the operating region $\hat{\cX}_0$ shown in Fig.~\ref{fig:cstr_state_space_h_1}, and a temporary target $\bar{y}_0$ is drawn uniformly from $\bar{\cY}_0 \doteq [2.4, 4.6]$.
The region $\hat{\cX}_0$ is chosen to contain the steady states associated with targets in $\bar{\cY}_0$.
We then run the initialization procedure, taking the initial state of the EKF as $\hat{x}_{-M|-M} = x_s$ and $P_{-M|-M} = I$, where $M = 30$ is the length of the initialization procedure.
We note that $M = 30$ is larger than the settling time of the closed-loop system~\eqref{eq:plant}-\eqref{eq:init_control}, meaning that when $k = 0$, the plant is close to the steady state corresponding to $\bar{y}_0$ and the estimated state $\hat{x}_{0|0}$ is a good estimate of the real system state $x_0$, as also considered in~\cite{KAL21}.

\subsection{Training the surrogate controller} \label{sec:results:training}
To capture local linear behavior around $(x_s, u_s)$ while~retaining expressive nonlinear corrections, $\hat f_z$, $\hat f_u$~\eqref{eq:surr_structure} are taken as the sum of an affine map and a Feedforward NN~(FNN):
\begin{subequations}
\begin{align}
    \hat f_z(v_k; \theta_z) &= W_z v_k + b_z + \psi_z(v_k), \label{eq:nn_fz1}\\
    \hat f_u(v_k; \theta_u) &= W_u v_k + b_u + \psi_u(v_k). \label{eq:nn_fu1}
\end{align}
\end{subequations}
where $W_z$, $b_z$, $W_u$, $b_u$ have appropriate dimensions, $\psi_z$ and $\psi_u$ are FNNs with hyperbolic tangent ($\tanh$) activation functions and three hidden layers of widths $(20,16,12)$ and $(16,12,8)$, respectively, and we recall that $v_k \doteq (z_k, y_k, \bar{y})$.

The training uses a dataset $\cD_t$ with $M_t = 100$ closed-loop trajectories of length $N_t=80$
\footnote{For informativeness, some initial states and regulation targets are sampled slightly outside $\hat{\cX}_0$ and $\bar{\cY}_0$, respectively, while remaining inside physically meaningful and safe operating bounds. This is useful from a practical point of view, as it improves performance close to the boundary of $\cX_0$ by training on points that are slightly outside of $\cX_0$.}.
Following Remark~\ref{rem:lyap_subset}, we greedily select $M_s=40$ trajectories and take $N_L=40$ for the Lyapunov regularization.
We note that $N_L = 40$ is longer than the settling time of the closed-loop system~\eqref{eq:plant}-\eqref{eq:expert}.

Problem~\eqref{eq:learn_prob_stabilize} is solved with \texttt{jax-sysid}~\cite{Bem25} using $6000$ Adam epochs followed by $4000$ L-BFGS epochs. The equilibrium-regularization weights are $\alpha_s^z=10^4$ and $\alpha_s^u=10^3$. The Lyapunov-regularization weight is $\alpha_{L}=10^4$, with $\underline{\mu}_V=\underline{\mu}_S=10^{-4}$, $\beta_x=10^{-3}$, and $\beta_z=10^{-4}$.
We take the parameter regularization term as 
\begin{align*}
r_\Theta(\Theta)
=
10^{-6}
( &
\|\theta\|_2^2 + \|z_s\|_2^2
+ \|\psi_V\|_F^2 + \|\psi_S\|_F^2 \\
      &+ |\mu_V|^2 + |\mu_S|^2 + |\tau|^2 ).
\end{align*}

For numerical conditioning during training, the temperature states are scaled as $\tilde{T} \doteq (T-300)/10$, while the concentration state is left unscaled.
The neural networks and Lyapunov functions use the scaled temperature, whereas the plant model, constraints, and trajectory plots are reported in unscaled units.

To reduce sensitivity to local minima in~\eqref{eq:learn_prob_stabilize}, we employ a parallel fitting approach where we train $8$ surrogate models from different random initializations, and select the one with the smallest closed-loop validation error over a set of $M_v = 8$ additional closed-loop trajectories of~\eqref{eq:plant}--\eqref{eq:expert}.
We define the validation error as
$J_{\mathrm{val}} = \frac{1}{M_v} \sum_{j=1}^{M_v} \frac{1}{N_t} \sum_{k=0}^{N_t-1} \| y_k^j - \hat{y}_k^j \|_2^2$.

The training is implemented in Python 3.12 and executed on an Intel(R) Xeon(R) W-2245 CPU @ 3.90GHz with 125 GB RAM.
Solving the offline learning problem~\eqref{eq:learn_prob_stabilize} required \SI{32}{\minute}.
A Python implementation of the proposed algorithms and numerical examples is available at \url{https://github.com/QuinnXie/CtrlFit}.

\begin{figure}
    \centering
    \includegraphics[width=1.0\columnwidth]{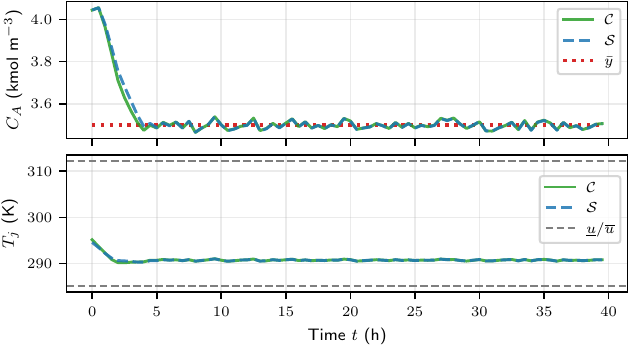} 
    \caption{Closed-loop trajectories of the expert controller $\mathcal{C}$~\eqref{eq:expert} (green lines) and the learned surrogate controller $\cS$~\eqref{eq:surr_control} (blue dashed lines) from the same initial condition with $C_{A,0} = 4.04$. Top: reactant concentration output $C_A$ and reference $\bar{y}$ (red dotted line). Bottom: jacket temperature input $T_j$ and input bounds (gray dashed lines).}
    \label{fig:cstr_stable_cmp}
\end{figure}

\subsection{Closed-loop performance and stability assessment} \label{sec:results:results}

Fig.~\ref{fig:cstr_stable_cmp} compares the closed-loop trajectories of the expert NMPC-EKF controller $\cC$ and the trained surrogate controller $\cS$, starting from the same initial condition and subject to the same sequence of measurement noise $\eta_k$. Both controllers successfully steer the system to the target equilibrium.
Additionally, the surrogate controller closely mimics the expert trajectory while reducing online computation time by more than 3000x relative to $\cC$ solved by \texttt{L-BFGS-B}, as shown in Table~\ref{tab:cstr_comp_time}.

\begin{table}[t]
\centering
\caption{CPU time for surrogate ($\cS$) and expert ($\cC$) controllers.}
\label{tab:cstr_comp_time}
\begin{tabular}{lcc}
\hline
Time [ms] & $\cS$ & $\cC$ \\
\hline
average & 0.0567 & 193.0 \\
maximum & 0.0678 & 210.7 \\
\hline
\end{tabular}
\end{table}

\begin{figure}[t]
    \centering
    \includegraphics[width=1.0\columnwidth]{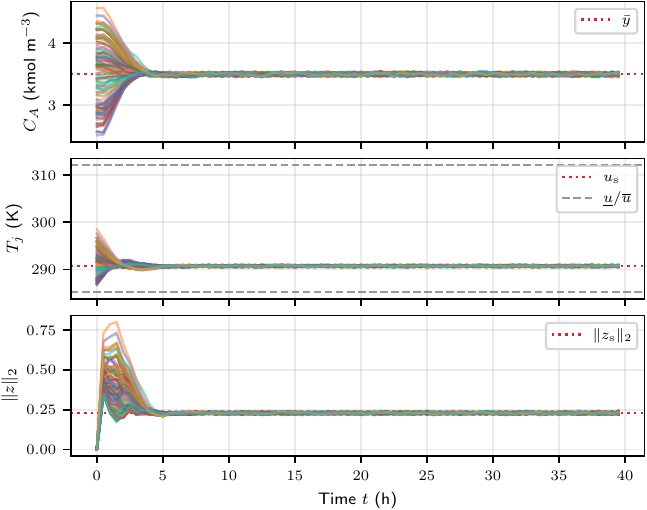} 
    \caption{Closed-loop trajectories of the learned surrogate controller from 100 initial states in $\cc{X}_0$. Top: reactant concentration output $C_A$ with reference $\bar{y}$ (red dotted line). Middle: jacket temperature input $T_j$, with the steady-state input $u_s$ (red dotted line) and the input bounds (gray dashed lines). Bottom: $\ell_2$-norm $\|z\|_2$ of the internal state $z$ of the surrogate controller $\cS$, with the steady-state value $\|z_s\|_2$ (red dotted line).}
    \label{fig:cstr_stable_eq}
\end{figure}

\begin{figure}[t]
    \centering
    \includegraphics[width=1.0\columnwidth]{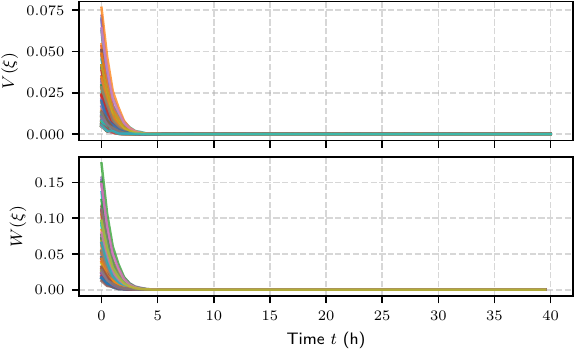}
    \caption{Evolution of the Lyapunov quantities along the closed-loop trajectories of the learned surrogate controller shown in Fig.~\ref{fig:cstr_stable_eq}. Top: Lyapunov function $V(\tilde{\xi})$. Bottom: auxiliary Lyapunov function $W(\tilde{\xi}; \eta)$, see~\eqref{eq:one_step_Lyap_W_noise}.}
    \label{fig:cstr_stable_v_w}
\end{figure}

Fig.~\ref{fig:cstr_stable_eq} illustrates the performance of $\cS$ starting from $100$ new random initial states, i.e., not used during training.
In all cases, $\cS$ drives the system to a neighborhood of the target equilibrium, illustrating that the learned surrogate controller provides practical stability.

Fig.~\ref{fig:cstr_stable_v_w} reports the evolution of the learned Lyapunov function $V(\tilde{\xi}_k)$ along the closed-loop trajectories shown in Fig.~\ref{fig:cstr_stable_eq}, with learned values: $\mu_V=4.80\times10^{-3}$, $\mu_S=1.96\times10^{-3}$, $\tau=0.81$.
Fig.~\ref{fig:cstr_stable_v_w} also shows the evolution of the auxiliary Lyapunov function $W(\tilde{\xi}_k; \eta_k)$ used in the proof of Theorem~\ref{theo:ISpS}.
While $V(\tilde{\xi}_k)$ may increase at isolated time steps under the two-step Lyapunov condition, $W(\tilde{\xi}_k; \eta_k)$ decreases monotonically during the transient. After around $t=\SI{5}{\hour}$, both functions remain near~$0$. This behavior is consistent with the ISpS result presented in Theorem~\ref{theo:ISpS}.

\subsection{Probabilistic validation} \label{sec:results:validation}
We now use the \emph{a posteriori} probabilistic validation approach from~\cite{KAL21} to show that constraint satisfaction and the ISpS condition~\eqref{eq:V_ISpS_cond} hold with high probability.
The procedure is based on measuring the values of two performance indicator functions $\phi_i$, $i \in \{L, \cX\}$ (see definitions below) for $N_p$ closed-loop experiments.
A probabilistic bound on the performance indicators can then be provided based on the empirical worst-case results obtained in the experiments.

Let $w_j$ be a random vector containing all the random variables that uniquely determine the outcome of the closed-loop experiment $j \in \N_1^{N_p}$.
In our setting, $w_j$ includes the random variables that affect the initialization procedure ($x_{-M}^j$, $\bar{y}_0^j$, etc.), as well as the values of the measurement noise $\eta_k^j$.
We denote by $\phi_i(w_j)$ the value of the performance indicator $i$ obtained in experiment~$j$.

We evaluate satisfaction of the ISpS condition~\eqref{eq:V_ISpS_cond} using
\begin{equation}\label{eq:prob_validation_phi_lyap}
    \phi_{L}(w_j) \doteq \max_{k \in \Ni{0}{N_t-2}} \left( \Gamma_k^j - \lambda_\eta\|(\eta_{k+1}^j, \eta_k^j)\| - c
    \right),
\end{equation}
where
$\Gamma_k^j \doteq \tau V(\tilde{\xi}_{k+2}^j) + V(\tilde{\xi}_{k+1}^j) - (1+\tau)V(\tilde{\xi}_{k}^j) + S(\tilde{\xi}_k^j)$.
Here, $\lambda_\eta > 0$ is the disturbance gain of the $\alpha_\eta$ function and $c \geq 0$ is the practical-stability residual.
Note that $\phi_{L}(w_j) \leq 0$ is equivalent to
$ \Gamma_k^j \leq \lambda_\eta\|(\eta_{k+1}^j, \eta_k^j)\| + c$, $\forall k \in \Ni{0}{N_t-2}$, i.e., the ISpS descent condition~\eqref{eq:V_ISpS_cond} holds throughout experiment~$j$.

We also evaluate constraint satisfaction by considering the performance indicator function
\begin{equation} \label{eq:prob_validation_phi_safe}
    \phi_{\cX}(w_j)
    \doteq
    \max_{k \in \Ni{0}{N_t}}
    \max_{i \in \Ni{1}{n_{\cX}}}
    \left(G_{\cX,i} x_k^j - h_{\cX,i}\right),
\end{equation}
where $n_\cX$ is the number of inequalities defining 
$\cX=\{x\in\R^{n_x}:G_\cX x\le h_\cX\}$, 
$G_{\cX,i}$ is the $i$-th row of $G_\cX$, and 
$h_{\cX,i}$ is the $i$-th entry of $h_\cX$.
Satisfaction of $\phi_\cX(w_j) \leq 0$ implies that the state constraints were satisfied throughout experiment $j$.

We take $\lambda_\eta=1$, $c=0$, and the parameters of the probabilistic validation procedure as $\epsilon = 10^{-4}$, $\delta = 10^{-6}$ and $r = 1$ (see~\cite[Section 3]{KAL21} for details), meaning that we must take $N_p \geq 145087$. 
We evaluated the performance indicators~\eqref{eq:prob_validation_phi_lyap} and~\eqref{eq:prob_validation_phi_safe} on $N_p=150000$ closed-loop experiments, using i.i.d.\ random variables $w_j$.
The maximum values of the performance indicators in the $N_p$ experiments are $\max_j \phi_{L}(w_j) = -2.0\times10^{-5}$ and $\max_j \phi_\cX(w_j) = -2.5$.
The probabilistic validation procedure from~\cite{KAL21} then gives the following result: with confidence at least ${1 - \delta} = 99.9999\%$, the values of $\phi_{L}(w)$ and $\phi_\cX(w)$ of a new experiment $w$ satisfy $\mathbb{P}\{\phi_{L}(w) {>} -2.0\times10^{-5}\} {\leq} 10^{-4}$ and $\mathbb{P}\{\phi_\cX(w) {>} -2.5\} {\leq} 10^{-4}$.
Equivalently, since the validated levels of $\phi_\cX$ and $\phi_{L}$ are negative and $c=0$, the learned surrogate controller provides constraint satisfaction and an ISpS closed-loop system with high probability.
This result highlights the effectiveness of the proposed method for all practical engineering purposes.

\begin{figure}
    \centering
    \includegraphics[width=1.0\columnwidth]{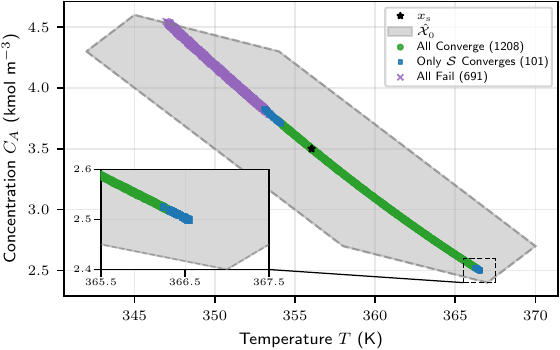} 
    \caption{Comparison of closed-loop stabilization outcomes for the short-horizon case $H=1$, evaluated from 2000 random initial states. The gray shaded polygon denotes the initial-state set $\hat \cX_0$ used in the initialization procedure, and the black star marks the target equilibrium state $x_s$. Green circles indicate initial conditions for which all controllers converge, blue squares indicate convergence of $\cS$ only, and purple crosses indicate initial conditions for which all controllers fail to track the reference.}
    \label{fig:cstr_state_space_h_1}
\end{figure}

\begin{figure}
    \centering
    \includegraphics[width=1.0\columnwidth]{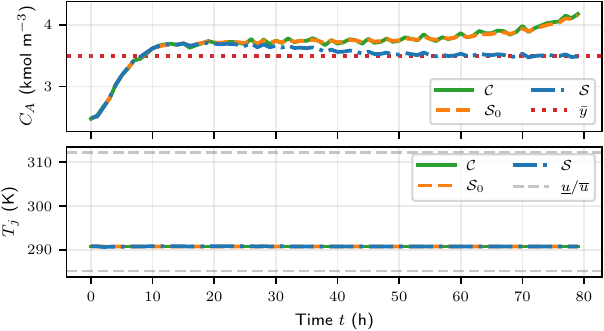} 
    \caption{Closed-loop comparison for the short-horizon case $H=1$ from one representative initial condition with $C_A=\SI{2.5}{\kilo\mol\per\meter\cubed}$. 
    $\mathcal{C}$ (green lines), $\cS_0$ (orange dashed lines), and $\cS$ (blue dash-dot lines) are compared. 
    Top: concentration output $C_A$, with the reference $\bar{y}$ (red dotted line). Bottom: jacket temperature input $T_j$, with the input bounds (gray dashed lines).}
    \label{fig:cstr_traj_h_1}
\end{figure}

\subsection{Role of the Lyapunov regularization}
To highlight the role of the equilibrium and Lyapunov regularization terms~$J_s$ and~$J_L$ of~\eqref{eq:learn_prob_stabilize}, we now consider the use of an expert controller~$\cC$ whose NMPC has a prediction horizon of $H = 1$.
We train the surrogate controller~$\cS$ using the same approach described in Section~\ref{sec:results:training}.
Additionally, we train a surrogate controller $\cS_0$ using the same procedure, but removing the regularization terms~$J_s$ and~$J_L$, i.e., we only include the imitation loss term $J(\cD_t; \theta)$ and the parameter regularization term $r_\Theta$.
We note that the initialization procedure still uses the NMPC with a prediction horizon of $H = 10$, to ensure proper initialization.

Fig.~\ref{fig:cstr_state_space_h_1} reports the comparison over 2000 random initial states.
All three controllers stabilize the system in 1208 cases; in 101 cases only $\cS$ succeeds; and all controllers fail in the remaining 691 cases.
We note that this last result is natural, given the short horizon $H = 1$.
The figure shows that the stabilization region of $\cS$ is slightly larger than that of the expert controller, while $\cS_0$ generally mimics the expert. 

Fig.~\ref{fig:cstr_traj_h_1} shows a representative closed-loop trajectory from an initial condition marked by a blue square in Fig.~\ref{fig:cstr_state_space_h_1}. Controllers $\cC$ and $\cS_0$ exhibit similar transient behavior, since $\cS_0$ is trained only by imitation, both failing to reach the target $\bar{y}$ within the finite simulation horizon.
On the other hand, $\cS$ converges to the target,
showing that the inclusion of the Lyapunov-related terms shapes the surrogate dynamics beyond pointwise imitation, towards practical stability.
This motivates the use of the proposed surrogate controller to imitate a well-performing expert controller-observer pair that is not necessarily stabilizing by design.

\section{Conclusions} \label{sec:conclusions}

We proposed a framework to learn a practically-stabilizing surrogate of a given nonlinear controller-observer design. The surrogate is represented as a recurrent output-feedback controller deployed using measured outputs and is trained by combining input/output imitation, equilibrium regularization, and Lyapunov-based regularization.
The nonlinear CSTR example showed that the learned surrogate can closely reproduce the expert NMPC-EKF behavior with much lower online complexity, while improving over an imitation-only baseline, and probabilistically providing constraint satisfaction and ISpS. The comparison with an imitation-only baseline indicates that the Lyapunov-based regularization can bias the learned surrogate toward practically stabilizing closed-loop behavior. Future work will address reference tracking, uncertainty-aware and adaptive extensions, and stronger safety and invariance certificates.

\bibliographystyle{ieeetr}
\bibliography{surrogate_stable_tac}

\end{document}